\documentclass[conference]{IEEEtran}
\usepackage{mathrsfs,amsmath,amssymb,epsfig,amsfonts,fancyhdr,graphicx,cite,amsthm}

\newtheorem{thm}{Theorem}[]
\newtheorem{cor}{Corollary}[]

\theoremstyle{remark}
\newtheorem{rem}[]{Remark}

\theoremstyle{definition}

\begin{document}

\title{Decode-forward and Compute-forward Coding Schemes for the
  Two-Way Relay Channel}
\author{\authorblockN{Peng Zhong and Mai Vu\\}
\authorblockA{Department of Electrical and Computer Engineering\\
McGill University\\
Montreal, QC, Canada H3A 2A7\\
Emails: peng.zhong@mail.mcgill.ca, mai.h.vu@mcgill.ca}}

\maketitle

\begin{abstract}
We consider the full-duplex two-way relay channel with direct link
between two users and propose two coding schemes: a partial
decode-forward scheme, and a combined decode-forward and
compute-forward scheme. Both schemes use rate-splitting and
superposition coding at each user and generate codewords for each node
independently. When applied to the Gaussian channel, partial
decode-forward can strictly increase the rate region over
decode-forward, which is opposite to the one-way relay channel. The
combined scheme uses superposition coding of both Gaussian and lattice
codes to allow the relay to decode the Gaussian parts and compute
the lattice parts. This scheme can also achieve new rates and outperform both decode-forward and
compute-forward separately. These schemes are steps towards
understanding the optimal coding.
\end{abstract}

\IEEEpeerreviewmaketitle

\section{Introduction}\label{sec:intro}
\IEEEPARstart{T}{he} two-way channel in which two users wish to
exchange message was first studied by Shannon \cite{shannon1961two}. A
specific model is the two-way relay channel (TWRC) with a relay located
between two users to help exchange messages. Two types of TWRC exist:
one without a direct link between the two users, a model suitable for
wired communication, and one with the direct link, more suitable for
wireless communication.
In this paper, we focus on the TWRC with direct link between the two
users, also called the full TWRC.

A number of coding schemes have been proposed for the full
TWRC. Different relay strategies, including amplify-and-forward,
decode-forward based on block Markov coding, compress-forward
and a combined decode-forward and compress-forward scheme, are studied
in \cite{rankov2006achievable}. For the decode-forward strategy, the
relay reliably decodes the transmitted messages from both users. It
then re-encodes and forwards. For the compress-forward strategy, the
relay compresses the noisy received signal and forwards. In
\cite{xie2007network}, a decode-forward scheme based on random binning and no
block Markovity was proposed, in which the relay broadcasts the bin
index of the decoded message pair.

A new relaying strategy called compute-forward was recently proposed
in \cite{nazer2009compute}, in which the relay decodes linear functions of
transmitted messages. Nested lattice code \cite{erez2005lattices} is
used to implement compute-forward in Gaussian channels, since it
ensures the sum of two  codewords is still a codeword. Compute-forward
has been shown to outperforms DF in moderate SNR regimes but is worse at low or
high SNR \cite{nazer2009compute}. Compute-forward can be naturally applied in two-way
relay channels as the relay now receives signal containing more than
one message. In \cite{narayanan2007joint}, nested lattice codes were
proposed for the Gaussian separated TWRC with symmetric channel,
i.e. all source and relay nodes have the same transmit powers and
noise variances. For the more general separated AWGN TWRC case,
compute-forward coding with nested lattice code can achievable rate
region within 1/2 bit of the cut-set outer bound
\cite{nam2008capacity}\cite{nam2009capacity}. For the full AWGN TWRC, a
scheme based on compute-forward, list decoding and random binning
technique is proposed in \cite{song2010list}. This scheme achieves rate region
within 1/2 bit of the cut-set bound in some cases.

In this paper, we consider the ideas of decode-forward and
compute-forward together and propose two new coding schemes for the
full TWRC. The first scheme is a partial decode-forward scheme which
extends the decode-forward scheme in \cite{xie2007network}. Each user splits
its message into two parts. The relay decodes one part of message from
each user, re-encode these two parts together and forwards. This
scheme contains the original decode-forward scheme in \cite{xie2007network} as a
special case. Different from the one-way relay channel in which
partial decode-forward brings no improvement on the achievable rate
over decode-forward in Gaussian channels \cite{cover1979capacity}, somewhat surprisingly here
for the full TWRC, partial decode-forward can achieve new rates and strictly increase the
rate region over decode-forward.

The second scheme combines decode-forward scheme with
compute-forward for the full Gaussian TWRC. Each user also splits its
message into two parts, and encodes one part with a Gaussian codeword
and the other with a lattice codeword. The relay chooses to
decode-forward one part of the message from each user, while
compute-forward the other part. This scheme can also achieve new rates and a better rate
region than either decode-forward and compute-forward alone.


\section{Channel Model}\label{sec:system_model}
\subsection{Discrete memoryless TWRC model}
The discrete memoryless two-way relay channel (DM-TWRC) is
denoted by $(\mathcal{X}_1 \times \mathcal{X}_2 \times
\mathcal{X}_r,p(y_1,y_2,y_r|x_1,x_2,x_r),\mathcal{Y}_1 \times
\mathcal{Y}_2 \times \mathcal{Y}_r)$, as in Figure \ref{p2}. Here
$x_1$ and $y_1$ are the input and output signals of user 1; $x_2$ and
$y_2$ are the input and output signals of user 2; $x_r$ and $y_r$ are
the input and output signals of the relay. We consider a full-duplex
channel in which all nodes can transmit and receive at the same time.

\begin{figure}[t]
\centering
  \includegraphics[scale=0.7]{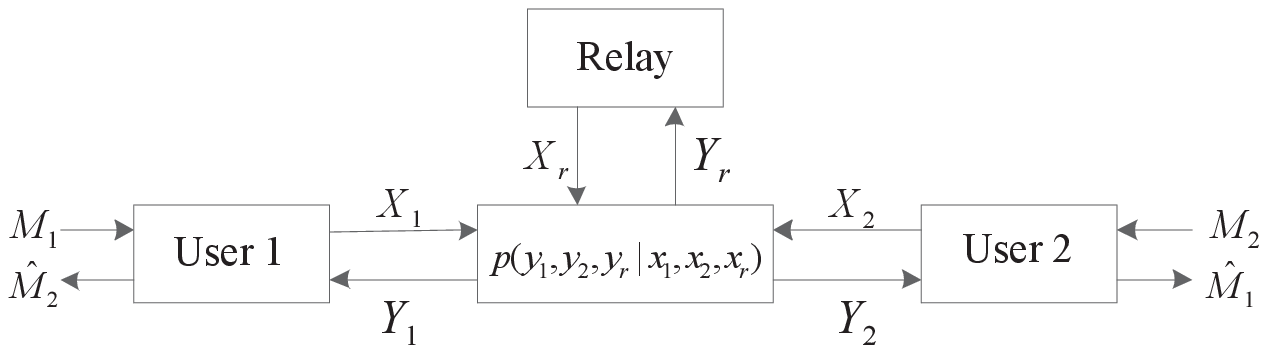}\\
  \caption{Two-way relay channel model}\label{p2}
\end{figure}

A $(n,2^{nR_1},2^{nR_2},P_e)$ code for a DM-TWRC consists of two
message sets $\mathcal{M}_1=[1:2^{nR_1}]$ and
$\mathcal{M}_2=[1:2^{nR_2}]$, three encoding functions
$f_{1,i},f_{2,i},f_{r,i}$, $i=1, \ldots, n$ and two decoding function
$g_1,g_2$.
\begin{align}
x_{1,i}&=f_{1,i}(M_1,Y_{1,1},\ldots ,Y_{1,i-1}),~~~~i=1, \ldots, n\nonumber\\
x_{2,i}&=f_{2,i}(M_2,Y_{2,1},\ldots ,Y_{2,i-1}),~~~~i=1, \ldots, n\nonumber\\
x_{r,i}&=f_{r,i}(Y_{r,1},\ldots ,Y_{r,i-1}),~~~~i=1, \ldots, n\nonumber\\
g_1&:\mathcal{Y}^n_1 \times \mathcal{M}_1 \rightarrow \mathcal{M}_2,~~~~g_2:\mathcal{Y}^n_2 \times \mathcal{M}_2 \rightarrow \mathcal{M}_1.\nonumber
\end{align}
The average error probability is
$P_e=\textrm{Pr}\{g_1(M_1,Y^n_1)\neq
M_2~\textrm{or}~g_2(M_2,Y^n_2)\neq M_1\}$. A rate pair is said to be
achievable if there exists a $(n,2^{nR_1},2^{nR_2},P_e)$ code such
that $P_e\rightarrow 0$ as $n\rightarrow \infty$. The closure of the
set of all achievable rates $(R_1,R_2)$ is the capacity region of the
two-way relay channel.

\subsection{Gaussian TWRC model}
The full additive white Gaussian noise (AWGN) two-way relay channel
can be modeled as below.
\begin{align}
  Y_1&=X_r+X_2+Z_1 \nonumber \\
  Y_2&=X_r+X_1+Z_2 \nonumber \\
  Y_r&=X_1+X_2+Z_r
  \label{GTWRC}
\end{align}
where the noises are independent: $Z_1\sim\mathcal{N}(0,N_1), Z_2\sim\mathcal{N}(0,N_2),
Z_r\sim\mathcal{N}(0,N_r)$. The average
input power constraints for user 1, user 2 and the relay are
$P_1,P_2,P_r$ respectively.

\section{A Partial Decode-Forward Scheme}\label{sec:partial_df}
In this section, we provide an achievable rate region for the TWRC
with a partial decode-forward scheme. Each user splits its message
into two parts and uses superposition coding to encode them. The
relay only decodes one message part of each user and re-encode the
decoded message pair together and broadcast. It can either re-encode
each message pair separately or divides these message pairs into lists
and only encodes the list index, which is similar to the binning
technique in \cite{xie2007network}. Both strategies achieve the same rate
region. The users then decode the message from each other by joint
typicality decoding of both the current and previous blocks.

\subsection{Achievable rate for the DM-TWRC}
\begin{thm}
\label{thm:par}
The following rate region is achievable for the two-way relay channel:
\begin{align}
  R_1\leq&\;\min \{I(U_1;Y_r|U_2,X_r)+I(X_1;Y_2|U_1,X_2,X_r),\nonumber \\
  &\;~~~~~~I(X_1,X_r;Y_2|X_2)\}\nonumber\\
  R_2\leq&\; \min \{I(U_2;Y_r|U_1,X_r)+I(X_2;Y_1|U_2,X_1,X_r),\nonumber \\
  &\;~~~~~~I(X_2,X_r;Y_1|X_1)\}\nonumber\\
  R_1+R_2\leq&\; I(U_1,U_2;Y_r|X_r)+I(X_1;Y_2|U_1,X_2,X_r)\nonumber \\
  &\;+I(X_2;Y_1|U_2,X_1,X_r)
\end{align}
for some joint distribution $p(u_1,x_1)p(u_2,x_2)p(x_r)$.
\end{thm}

\begin{rem}
If $U_1=X_1,U_2=X_2$, this region reduces to the
decode-forward lower bound in \cite{xie2007network}. Therefore, the partial DF
scheme contains the DF scheme in \cite{xie2007network} as a special case.
\end{rem}

\begin{figure*}[t]
  \begin{align}
    R_1 &\leq \min \left\{C\left(\frac{\alpha
      P_1}{\bar{\alpha}P_1+\bar{\beta}P_2+N_r}\right)
    +C\left(\frac{\bar{\alpha}P_1}{N_2}\right),
    C\left(\frac{P_1+P_r}{N_2}\right)\right\}\nonumber\\
    R_2 &\leq \min \left\{C\left(\frac{\beta P_2}
    {\bar{\alpha}P_1+\bar{\beta}P_2+N_r}\right)+
    C\left(\frac{\bar{\beta}P_2}{N_1}\right),
    C\left(\frac{P_2+P_r}{N_1}\right)\right\}\nonumber \\
    R_1+R_2 & \leq C\left(\frac{\alpha P_1+\beta P_2}
    {\bar{\alpha}P_1+\bar{\beta}P_2+N_r}\right)+
    C\left(\frac{\bar{\alpha}P_1}{N_2}\right)
    +C\left(\frac{\bar{\beta}P_2}{N_1}\right), \quad \text{where }
    0\leq \alpha, \beta \leq 1.
    \label{gtwrc-pdf}
  \end{align}
\end{figure*}

\begin{proof}
We use a block coding scheme in which each user sends $B-1$ messages
over $B$ blocks of $n$ symbols each.
\subsubsection{Codebook generation}
Fix $p(u_1,x_1)p(u_2,x_2)p(x_r)$.
Split each message into two parts: $m_1 = (m_{10},m_{11})$ with
rate $(R_{10},R_{11})$, and $m_2 = (m_{20},m_{22})$ with
rate $(R_{20},R_{22})$.
\begin{itemize}
\item Generate $2^{nR_{10}}$ i.i.d. sequences $u^n_{1}(m_{10})\sim
  \prod ^n_{i=1}p(u_{1i})$, where $m_{10} \in [1:2^{nR_{10}}]$. For
  each $u^n_{1}(m_{10})$, generate $2^{nR_{11}}$ i.i.d. sequences
  $x^n_{1}(m_{11},m_{10})\sim \prod ^n_{i=1}p(x_{1i}|u_{1i})$, where
  $m_{11} \in [1:2^{nR_{11}}]$.
\item Generate $2^{nR_{20}}$ i.i.d. sequences $u^n_{2}(m_{20})\sim
  \prod ^n_{i=1}p(u_{2i})$, where $m_{20} \in [1:2^{nR_{20}}]$. For
  each $u^n_{2}(m_{20})$, generate $2^{nR_{22}}$ i.i.d. sequences
  $x^n_{2}(m_{22},m_{20})\sim \prod ^n_{i=1}p(x_{2i}|u_{2i})$, where
  $m_{22} \in [1:2^{nR_{22}}]$.
\item Uniformly throw each message pair $(m_{10},m_{20})$ into $2^{nR_r}$
  bins. Let $K(m_{10},m_{20})$ denote the index of bin.

\item Generate $2^{nR_r}$ i.i.d. sequences $x^n_{r}(K)\sim \prod
  ^n_{i=1}p(x_{ri})$, where $K\in [1:2^{nR_r}]$. If $R_r=R_{10}+R_{20}$, there is no need for binning.

\end{itemize}
The codebook is revealed to all parties.

\subsubsection{Encoding}
In each block $b\in [1:B-1]$, user 1 and user 2 transmit
$x^n_{1}(m_{11}(b),m_{10}(b))$ and $x^n_{2}(m_{22}(b),m_{20}(b))$
respectively. In block $B$, user 1 and user 2 transmit $x^n_{1}(1,1)$
and $x^n_{2}(1,1)$, respectively.

At the end of block $b$, the relay has an estimate
$(\tilde{m}_{10}(b),\tilde{m}_{20}(b))$ from the decoding
procedure. It transmits
$x^n_r(K(\tilde{m}_{10}(b),\tilde{m}_{20}(b)))$ in block $b+1$.

\subsubsection{Decoding}
We explain the decoding strategy at the end of block $b$.
\subsubsection*{Decoding at the relay}
Upon receiving $y_r^n(b)$, the relay searches for the unique pair
$(\tilde{m}_{10}(b),\tilde{m}_{20}(b))$ such that
\begin{align*}
  \bigl( & u_1^n(\tilde{m}_{10}(b)),u_2^n(\tilde{m}_{20}(b)),y_r^n(b),\\
  & x^n_r\left(K(\tilde{m}_{10}(b-1),\tilde{m}_{20}(b-1))\right)\bigr)
  \in A^n_{\epsilon}.
\end{align*}
Following the analysis in multiple access channel, the error
probability will go to zero as $n\rightarrow \infty$ if
\begin{align}
  R_{10}&\;\leq I(U_1;Y_r|U_2,X_r)\nonumber \\
  R_{20}&\;\leq I(U_2;Y_r|U_1,X_r)\nonumber \\
  R_{10}+R_{20}&\;\leq I(U_1,U_2;Y_r|X_r).
\label{par1}
\end{align}

\subsubsection*{Decoding at each user}
By block $b$, user 2 has decoded $m_1(b-2)$. At the end of block $b$,
it searches for a unique message pair $(\hat{m}_{10}(b-1),\hat{m}_{11}(b-1))$
such that
\begin{align*}
  \big(x^n_r(K(\hat{m}_{10}(b-1),{m}_{20}(b-1))),y^n_2(b),x^n_2(b)\big)&
  \in A^n_{\epsilon}\\
  \textrm{and}~~~~\big(u_1^n(\hat{m}_{10}(b-1)),x^n_1(\hat{m}_{11}(b-1),
  \hat{m}_{10}(b-1)),&\nonumber\\
  y^n_2(b-1),x^n_r(K(m_1(b-2),{m}_2(b-2))),x^n_2(b-1)\big)&\in A^n_{\epsilon}.
\end{align*}
Following joint decoding analysis, the error probability will go to
zero as $n\rightarrow \infty$ if
\begin{align}
  R_{11}&\;\leq I(X_1;Y_2|U_1,X_2,X_r)\nonumber\\
  R_{10}+R_{11}&\;\leq I(X_r;Y_2|X_2)+I(U_1,X_1;Y_2|X_2,X_r)\nonumber\\
  &\;=I(X_1,X_r;Y_2|X_2).
\end{align}

Similarly, user 1 can decode $({m}_{20}(b-1),{m}_{22}(b-1))$ with
error probability goes to zero as $n\rightarrow \infty$ if
\begin{align}
  R_{22}&\;\leq I(X_2;Y_1|U_2,X_1,X_r)\nonumber \\
  R_{20}+R_{22}&\;\leq I(X_2,X_r;Y_1|X_1).
  \label{par2}
\end{align}

By applying Fourier-Motzkin Elimination to the inequalities in
(\ref{par1})-(\ref{par2}), the achievable rates in terms of
$R_1=R_{10}+R_{11}$ and $R_2=R_{20}+R_{22}$ are as given in Theorem
\ref{thm:par}.
\end{proof}

\subsection{Rate region for the Gaussian TWRC}
Now we apply the proposed partial decode-forward scheme to the AWGN
TWRC in \eqref{GTWRC}. Using jointly Gaussian codewords, we can derive
an achievable rate region as follows.
\begin{cor}
  The rate region in \eqref{gtwrc-pdf} is achievable for the AWGN
  two-way relay channel.
\end{cor}

Achievability follows from Theorem \ref{thm:par} by setting
$X_1=U_1+V_1$, where $U_1\sim\mathcal{N}(0,\alpha P_1)$ and
$V_1\sim\mathcal{N}(0,\bar{\alpha} P_1)$ are independent, and by
setting $X_2=U_2+V_2$, where $U_2\sim\mathcal{N}(0,\beta P_2)$ and
$V_2\sim\mathcal{N}(0,\bar{\beta} P_2)$ are independent.

\begin{cor}
Partial decode-forward achieves strictly better region region than the decode-forward scheme in \cite{xie2007network} when
the following condition holds:
\begin{align}
N_r&>\min\{N_1,N_2\}\nonumber\\
\textrm{or}~~~C(P_1/N_2)+C(P_2/N_1)&>C((P_1+P_2)/N_r).
\end{align}
\end{cor}

The larger rate region of partial decode-forward can come from time sharing of decode-forward and direct transmission (without using the relay). But for asymmetric channels, new rates outside this time-shared region are also achievable as shown in the numerical results section.

\section{A Combined Decode-Forward and Compute-Forward Scheme for
  the Gaussian TWRC}\label{sec:comb}
In this section, we propose a combined decode-forward and
compute-forward scheme and analyze its rate regions for the Gaussian
TWRC. Each user split its message into two parts. One part is encoded
by a random Gaussian code, while another part is encoded by a lattice
code. The user transmits a superposition codeword of these two
parts. The relay decodes the Gaussian codewords of both users and a
function (the sum) of the two lattice codewords. It then jointly
encodes all 3 decoded parts and forwards. Again the relay can assign a
separate codeword to each set of the 3 decoded parts or it can encode only
the list index as in \cite{xie2007network} without affecting the achievable
rate. The users apply both joint typicality decoding and list
lattice decoding \cite{song2010list} to decode the message from each other.
The combined scheme achieves genuinely new rate. An example will be given in the numerical
result section.

\begin{thm}
\label{thm2}
The following rate region is achievable for the AWGN two-way relay channel:
\begin{align}
\label{combined_rate}
  R_{10}&\;\leq C\left(\frac{\alpha P_1}
  {\bar{\alpha} P_1+\bar{\beta} P_2+N_r}\right)=I_1\nonumber \\
  R_{20}&\;\leq C\left(\frac{\beta P_2}
  {\bar{\alpha} P_1+\bar{\beta} P_2+N_r}\right)=I_2\nonumber\\
  R_{10}+R_{20}&\;\leq C\left(\frac{\alpha P_1+\beta P_2}
  {\bar{\alpha} P_1+\bar{\beta} P_2+N_r}\right)=I_3\nonumber\\
  R_{11}&\;< \frac{1}{2}\log \left(\frac{\bar{\alpha} P_1}
  {\bar{\alpha} P_1+\bar{\beta} P_2}+\frac{\bar{\alpha} P_1}{N_r}\right)^+=I_4\nonumber\\
  R_{22}&\;< \frac{1}{2}\log \left(\frac{\bar{\beta} P_2}
  {\bar{\alpha} P_1+\bar{\beta} P_2}+\frac{\bar{\beta} P_2}{N_r}\right)^+=I_5\nonumber\\
  R_{10}&\;\leq C\left(\frac{\alpha P_1+\gamma P_r}
  {\bar{\alpha} P_1+\bar{\gamma} P_r+N_2}\right)=I_6\nonumber\\
  R_{20}&\;\leq C\left(\frac{\beta P_2+\gamma P_r}
  {\bar{\beta} P_2+\bar{\gamma} P_r+N_1}\right)=I_7\nonumber\\
  R_{11}&\;\leq C\left(\frac{\bar{\gamma} P_r}{P_1+N_2}\right)
  +C\left(\frac{\bar{\alpha} P_1}{N_2}\right)=I_8\nonumber\\
  R_{22}&\;\leq C\left(\frac{\bar{\gamma} P_r}{P_2+N_1}\right)
  +C\left(\frac{\bar{\beta} P_2}{N_1}\right)=I_9
\end{align}
where $0\leq {\alpha}, {\beta}, {\gamma} \leq 1$ and $[x]^+\triangleq \max\{x,0\}$. By applying
Fourier-Motzkin Elimination to the above inequalities, the achievable
rates in terms of $R_1=R_{10}+R_{11}$ and $R_2=R_{20}+R_{22}$ can be
expressed as
\begin{align}
  R_1&\;\leq \min (I_1,I_6)+\min (I_4,I_8)\nonumber\\
  R_2&\;\leq \min (I_2,I_7)+\min (I_5,I_9)\nonumber\\
  R_1+R_2&\;\leq I_3+\min (I_4,I_8)+\min (I_5,I_9).
\end{align}
\end{thm}

\begin{proof}
We use block coding scheme in which each user sends $B-1$ messages
over $B$ blocks of $n$ symbols.
\subsubsection{Codebook generation}
Let $P_1=\alpha P_1+\bar{\alpha} P_1$ and $P_2=\beta P_2+\bar{\beta}
P_2$. Without loss of generality, assume $\bar{\alpha} P_1\geq \bar{\beta} P_2$,
construct a chain of nested lattices
$\Lambda_1\subseteq\Lambda_2\subseteq\Lambda_{c1}\subseteq\Lambda_{c2}$,
where $\sigma^2(\Lambda_1)=\bar{\alpha} P_1$ and $\sigma^2(\Lambda_2)=\bar{\beta}
P_2$. $\Lambda_1$ and $\Lambda_2$ are Rogers-good and Poltyrev-good,
while $\Lambda_{c1}$ and $\Lambda_{c2}$ are Poltyrev-good
\cite{erez2004achieving,erez2005lattices}.

Split each message into two parts: $m_1 = (m_{10},m_{11})$
with rate $(R_{10},R_{11})$, and $m_2 = (m_{20},m_{22})$ with
rate $(R_{20},R_{22})$.


\begin{itemize}
\item Generate $2^{nR_{10}}$ random Gaussian codewords
  $u^n_1(m_{10})$ with power
  constraint $\alpha P_1$. Associate each message
  $m_{11}\in[1:2^{nR_{11}}]$ with lattice codeword $t^n_1\in
  \mathcal{C}_1=\Lambda_{c1}\cap\mathcal{V}_1$. Let
  $v^n_1(m_{11})=(t^n_1(m_{11})+U^n_1(m_{11}))\mod \Lambda_1$, where
  $U^n_1$ is the uniformly generated dither sequence known to all
  users and the relay. The codeword for $m_1$ is a superposition of
  the random Gaussian code and the lattice code:
  $x^n_1(m_1)=u^n_1(m_{10})+v^n_1(m_{11})$.
\item Similarly generate
  $2^{nR_{20}}$ random Gaussian codewords $u^n_2(m_{20})$ with power
  constraint $\beta P_2$, and $2^{nR_{22}}$ lattice codewords  $t^n_2(m_{22})$.
Let $x^n_2(m_2)=u^n_2(m_{20})+v^n_2(m_{22})$.
\item Uniformly throw each pair $(m_{10},m_{20})$ into $2^{nR_{r0}}$
  bins. Let $K((m_{10},m_{20}))$ denotes the bin index.
\item Form the computed codewords
  $T^n=(t^n_1(m_{11})+t^n_2(m_{22})-Q_2(t^n_2(m_{22})+U^n_2(m_{22})))\mod
  \Lambda_1$, where $Q_2(t^n)$ is
  the lattice quantizer mapping $t^n$ to the nearest lattice
  point. Uniformly throw $T^n$ into $2^{nR_{r1}}$ bins. Let
  $S(T^n)$ denotes the bin index.
\item Generate $2^{nR_{r0}}$ Gaussian codewords $u^n_r(K)$ with
  power constraint $\gamma P_r$ and $2^{nR_{r1}}$
  Gaussian codewords $v^n_r(S)$ with power constraint
  ${\bar{\gamma}}P_r$. Let $x^n_r=u^n_r(K)+v^n_r(S)$.
\end{itemize}
The codebook is revealed to all nodes.

\subsubsection{Encoding}
In block $b$, user 1 sends $x^n_1(m_1(b))$ and user 2 sends
$x^n_2(m_2(b))$. Assume the relay has decoded
$(m_{10}(b-1),m_{20}(b-1))$ and $T^n(b-1)$ in block $b-1$. It then
sends
$x^n_r(b)=u^n_r(K(m_{10}(b-1),m_{20}(b-1)))+v^n_r(S(T^n(b-1)))$ in
block $b$.

\subsubsection{Decoding}
We explain the decoding strategy at the end of block $b$.
\subsubsection*{Decoding at the relay}
The relay first decodes $m_{10}(b)$ and $m_{20}(b)$ using joint
typicality decoding. Similar to the analysis in multiple access
channel, $P_e\to 0$ as $n\to\infty$ if
\begin{align}
  R_{10}\leq I(U_1;Y_r|U_2,X_r)\nonumber \\
  R_{20}\leq I(U_2;Y_r|U_1,X_r)\nonumber \\
  R_{10}+R_{20}\leq I(U_1,U_2;Y_r|X_r).
  \label{com1}
\end{align}\par
The relay then subtracts $u^n_1(m_{10}(b))$ and
$u^n_2(m_{20}(b))$ from its received signal. Following arguments
similar to those in \cite{nazer2009compute} \cite{nam2009capacity}, it
can then decode $T^n(b)$ with vanishing error as long as
\begin{align}
R_{11}&\;\leq \frac{1}{2}\log \left(\frac{\bar{\alpha} P_1}{\bar{\alpha} P_1+\bar{\beta}
  P_2}+\frac{\bar{\alpha} P_1}{N_r}\right)\nonumber \\
R_{22}&\;\leq \frac{1}{2}\log \left(\frac{\bar{\beta} P_2}{\bar{\alpha} P_1+\bar{\beta}
  P_2}+\frac{\bar{\beta} P_2}{N_r}\right).
\end{align}

\subsubsection*{Decoding at each user}
At the end of block $b$, user 2 first decodes the unique $m_{10}(b-1)$
such that
\begin{align*}
\big(u^n_r(K(m_{10}(b-1),m_{20}(b-1))),y^n_2(b),x^n_2(b)\big)&\in A^n_\epsilon\\
\textrm{and}~~\big(u^n_1(m_{10}(b-1)),u^n_r(K(m_{10}(b-2),m_{20}(b-2))),&\\
x^n_2(b-1),y^n_2(b-1)\big)& \in A^n_\epsilon.
\end{align*}
This decoding has vanishing error probability if
\begin{align}
  R_{10}&\;\leq I(U_r;Y_2|X_2)+I(U_1;Y_2|U_r,X_2)\nonumber \\
  &\;=I(U_1,U_r;Y_2|X_2).
\end{align}

User 2 then subtracts $u^n_1(m_{10}(b-1))$ from $y^n_2(b-1)$ and
uses a lattice list decoder \cite{song2010list} to decode a list of possible
$m_{11}(b-1)$ of size $2^{n(R_{11}-C\left(\bar{\alpha} P_1/N_2)\right)}$, denoted
as $L(m_{11}(b-1))$. To decode which message in this list was sent, it uses the received signal in block $b$. That is to say, it
decodes the unique $m_{11}(b-1)$ such that
\begin{align*}
  \big(y_2^n(b),x_2^n(b), u^n_r(K(m_{10}(b-1),m_{20}(b-1))), &\\
  x^n_r(K(m_{10}(b-1),m_{20}(b-1)),S(T^n(b-1)))\big) &\in A^n_\epsilon\\
  \textrm{and}~~~~~~~~~~~~~~~~~~~~~~m_{11}(b-1) \in L(m_{11}(b&-1)).
\end{align*}
This decoding has vanishing error probability if
\begin{equation}
R_{11}\leq I(X_r;Y_2|X_2,U_r)+C\left(\bar{\alpha} P_1/N_2\right).
\end{equation}

Similarly, user 1 can decode $m_{20}(b-1),m_{22}(b-1)$ with vanishing
error as long as
\begin{align}
  R_{20}&\;\leq I(U_2,U_r;Y_1|X_1)\nonumber \\
  R_{22}&\;\leq I(X_r;Y_1|X_1,U_r)+C\left(\bar{\beta} P_2/N_1\right).
  \label{com2}
\end{align}

Finally, by setting
\begin{align*}
X_1&\;=U_1+V_1;~~~U_1\sim\mathcal{N}(0,\alpha P_1),\;
V_1\sim\mathcal{N}(0,\bar{\alpha} P_1)\\
X_2&\;=U_2+V_2;~~~U_2\sim\mathcal{N}(0,\beta P_2),\;
V_2\sim\mathcal{N}(0,\bar{\beta} P_2)\\
X_r&\;=U_r+V_r;~~~U_r\sim\mathcal{N}(0,\gamma P_r),\;
V_r\sim\mathcal{N}(0,\bar{\gamma} P_r)
\end{align*}
the achievable rate region in Theorem \ref{thm2} can be derived from
inequalities (\ref{com1})-(\ref{com2}).
\end{proof}

\section{Numerical Results}\label{Res}

In this section, we compare the achievable rate regions of the two
proposed schemes with pure decode-forward (DF) \cite{xie2007network} and pure
compute-forward \cite{song2010list}.\par
Figures \ref{s1} and \ref{s2} show
the achievable rate regions of pure DF \cite{xie2007network}, of direct transmission (without using the relay) and of the proposed
partial DF for 2 different channel configurations. Figure \ref{s1} shows that partial DF can achieve new rates outside the time sharing region of pure DF and direct transmission. For example, by setting $\alpha=1,\beta=0.5$ in (\ref{gtwrc-pdf}), partial DF can achieve the rate $(R_1,R_2)=(0.58,1.47)$ which is outside the convex hull of direct transmission and pure DF. This is notably different from
the one-way relay Gaussian channel in which partial DF brings no
improvement. In Figure \ref{s1}, the
channel from the users to the relay is stronger than the channel
between two users, thus the relay chooses partial DF to obtain a better
rate region than DF. Figure \ref{s2} shows performance for another channel configuration which is symmetric. In this case, the channels from two users to the relay are significantly stronger than the direct channel, and the relay will fully decode the messages.

Figures \ref{s3} and \ref{s4} present the achievable rate
regions for DF, compute-forward and the combined scheme. The cut-set outer bound is obtained
by assuming correlated channel inputs and independent channel noise, which is different from the cut-set bound in \cite{song2010list} for physically degraded
channels. Both Figures \ref{s3} and \ref{s4} show that the combined
scheme can achieve a better rate region than either DF and
compute-forward alone. Figure \ref{s3} for an asymmetric channel shows new rates outside the time-shared region of DF and compute-forward. For example, by setting $\alpha=0.5,\beta=0$ in (\ref{combined_rate}), the combined scheme can achieve the rate $(R_1,R_2)=(0.678,0.859)$ which is outside the convex hull of pure DF and pure compute-forward. In Figure \ref{s4} for a symmetric channel, the combined scheme can also achieve a boundary point of
$(R_1,R_2)=(0.69,1.01)$ by setting $\alpha=0.48,\beta=0$ in (\ref{combined_rate}), instead of time sharing of the two independent
schemes. These simulation results show that both proposed schemes achieve strictly new rates particularly for asymmetric channels. But because of space limitation, analyses of these new rates are left for future work.

\begin{figure}[t]
\centering
  \includegraphics[scale=0.5]{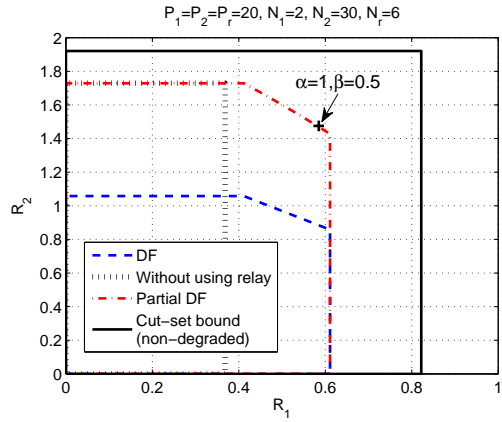}\\
  \caption{Achievable rate regions for DF and partial DF schemes with
    $P_1=P_2=P_r=20,N_1=2,N_2=30,N_r=6$}\label{s1}
\end{figure}

\begin{figure}[t]
\centering
  \includegraphics[scale=0.5]{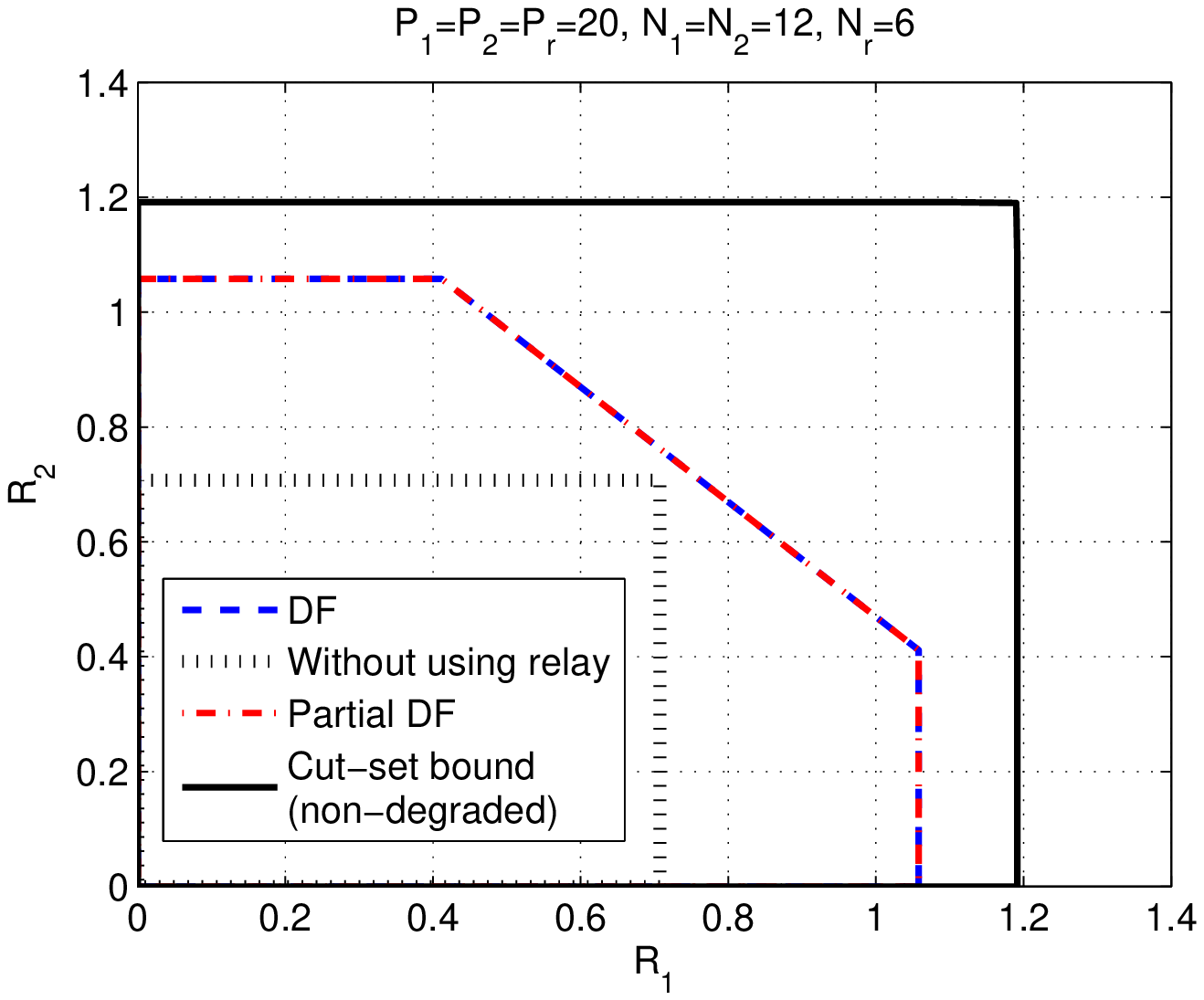}\\
  \caption{Achievable rate regions for DF and partial DF schemes with
    $P_1=P_2=P_r=20,N_1=N_2=12,N_r=6$}\label{s2}
\end{figure}

\begin{figure}[t]
\centering
  \includegraphics[scale=0.5]{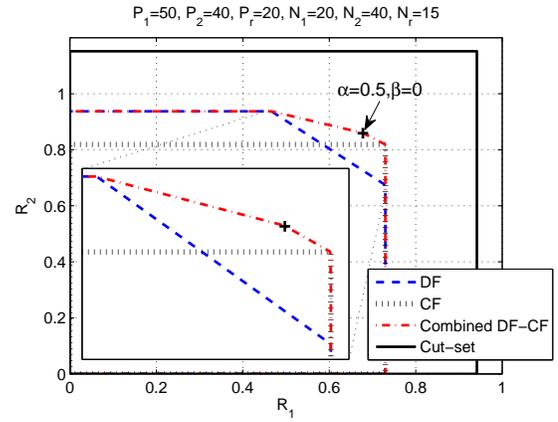}\\
  \caption{Achievable rate regions for DF, compute-forward, combined
    of DF and compute-forward, and partial DF schemes with
    $P_1=50, P_2=40, P_r=20,N_1=20, N_2=40 ,N_r=15$}\label{s3}
\end{figure}

\begin{figure}[t]
\centering
  \includegraphics[scale=0.5]{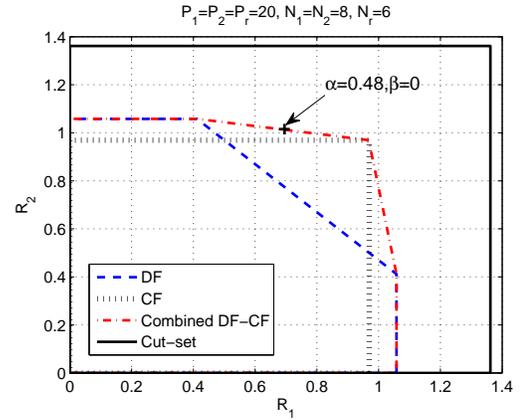}\\
  \caption{Achievable rate regions for DF, compute-forward, combined
    of DF and compute-forward, and partial DF schemes with
    $P_1=P_2=P_r=20,N_1=N_2=8,N_r=6$}\label{s4}
\end{figure}

\section{Conclusion}\label{Con}
We have proposed two new coding schemes for the two-way relay channel:
a partial decode-forward scheme and a combined decode-forward and
compute-forward scheme. Analysis for the Gaussian channel shows that
partial decode-forward can strictly increase the rate region of the
TWRC over pure decode-forward. This result is opposite to the one-way Gaussian
relay channel. In addition, combining decode-forward with
compute-forward by rate splitting and superposition of both Gaussian
and lattice codes can strictly outperform each separate scheme. These results
suggest more comprehensive coding schemes possible for the TWRC.

\bibliographystyle{IEEEtran}
\bibliography{reflist}

\end{document}